\documentclass[onecolumn,noversion,nonote]{cdcarticle}

\usepackage{amsfonts,amssymb,amsmath,amsthm}
\usepackage{enumerate,graphicx,tikz}
\usepackage[hidelinks]{hyperref}

\newtheorem{thm}{Theorem}
\newtheorem{lem}[thm]{Lemma}
\newtheorem{cor}[thm]{Corollary}
\newtheorem{defn}[thm]{Definition}
\newtheorem{prop}[thm]{Proposition}
\newtheorem{rem}[thm]{Remark}
\renewenvironment{proof}{{\noindent\bf Proof.}}{ \hfill ~\qed}
\def\qed{\rule[0pt]{5pt}{5pt}\par\medskip}

\newcommand{\neset}{\Omega}
\newcommand{\R}{\mathbb{R}}

\newcommand{\Rp}{\mathcal{R}_p}

\newcommand{\RHinf}{\mathcal{RH}_\infty}
\newcommand{\Htwo}{\mathcal{H}_2}

\DeclareMathOperator{\diag}{\mathbf{diag}}

\newcommand{\T}{\rule{0pt}{2.3ex}}
\newcommand{\hlinet}{\hline\T}
\newcommand{\stsp}[4]{\left[\begin{array}{c|c}
      #1 & #2 \\ \hlinet
      #3 & #4
    \end{array}\right]}
\newcommand{\bmat}[1]{\begin{bmatrix}#1\end{bmatrix}}

\begin{document}
\title{On Structured Realizability and \\ Stabilizability of Linear Systems}
\author{Laurent Lessard \and Maxim Kristalny \and Anders Rantzer}
\note{}
\maketitle

\begin{abstract}

We study the notion of \emph{structured realizability} for linear systems defined over graphs. A stabilizable and detectable realization is \emph{structured} if the state-space matrices inherit the sparsity pattern of the adjacency matrix of the associated graph. In this paper, we demonstrate that not every structured transfer matrix has a structured realization and we reveal the practical meaning of this fact.
We also uncover a close connection between the structured realizability of a plant and whether the plant can be stabilized by a structured controller. In particular, we show that a structured stabilizing  controller can only exist when the plant admits a structured realization. Finally, we give a parameterization of all structured stabilizing controllers and show that they always have structured realizations.
\end{abstract}

\section{Introduction}\label{sec:intro}

Linear time-invariant systems are typically represented using transfer functions or state-space realizations. The relationship between these representations is well understood; every transfer function has a minimal state-space realization and one can easily move between representations depending on the need.

In this paper, we address the question of whether this relationship between representations still holds for systems defined over directed graphs. Consider the simple two-node graph of Figure~\ref{fig:ex0}.
For $i=1,2$, the node $i$ represents a subsystem with inputs $u_i$ and outputs $y_i$, and the edge means that subsystem~1 can influence subsystem~2 but not vice versa. The transfer matrix for any such a system has the sparsity pattern $\mathcal{S}_1$.

For a state-space realization to make sense in this context, every state should be associated with a subsystem, which means that the state should be computable using the information available to that subsystem. In the case of $\mathcal{S}_1$, the states that determine $y_1$ must only depend on $u_1$, while the states that determine $y_2$ may depend on both $u_1$ and $u_2$. Consider the following example, which belongs to the graph constraint of Figure~\ref{fig:ex0}.

\begin{figure}[h]
\begin{equation}\label{eq:ex0}
G_1 = \bmat{ \frac{1}{s+1} & 0 \\ \T \frac{1}{s+1} & \frac{1}{s+2} }
= \left[\begin{array}{cc|cc}
-1 & 0 & 1 & 0 \\
0 & -2 & 0 & 1 \\ \hlinet
1 & 0 & 0 & 0  \\
1 & 1 & 0 & 0
\end{array}\right]
\end{equation}
\centering
\begin{minipage}{0.3\linewidth}
\begin{center}
\begin{tikzpicture}[thick,>=latex]
\tikzstyle{block}=[circle,draw,minimum height=2.1em]
\node [block](N1){$1$};
\path (N1) +(1.6,0) node [block](N2){$2$};
\draw [->] (N1) -- (N2);
\end{tikzpicture}
\end{center}
\end{minipage}
\begin{minipage}{0.3\linewidth}
\begin{center}
$\mathcal{S}_1 = \bmat{1&0\\1&1}$
\end{center}
\end{minipage}
\caption{A two-node graph and its adjacency matrix.\label{fig:ex0}}
\end{figure}
%
More generally, we say a system is $\mathcal{S}$\emph{-realizable} if it has a stabilizable and detectable realization for which $A,B,C,D$ each have the sparsity of $\mathcal{S}$. A more formal definition is given in Section~\ref{sec:realizability}.

While this definition seems natural and one might expect every structured transfer function to have a corresponding $\mathcal{S}$-realization, it is not the case in general. For example, $G_2$ defined in~\eqref{eq:ex1} belongs to the graph constraint of Figure~\ref{fig:ex1}, but no $\mathcal{S}_2$-realization exists. A proof is given in Appendix~\ref{sec:appendix}. This example shows that there is no immediate relation between the sparsity of a transfer matrix and that of its state-space matrices. 
\begin{figure}[h]
\begin{equation}\label{eq:ex1}
G_2 = \bmat{ 0 & 0 & 0 & 0 \\ 0 & 0 & 0 & 0 \\
\tfrac{1}{s-1} & \tfrac{1}{s-1} & 0 & 0 \\
\tfrac{1}{s-1} & \tfrac{1}{s-1} & 0 & 0 }
\end{equation}
\centering
\begin{minipage}{0.3\linewidth}
\begin{center}
\begin{tikzpicture}[thick,>=latex]
\tikzstyle{block}=[circle,draw,minimum height=2.1em]
\node [block](N3){$3$};
\path (N3) +(-1.3,-.6) node [block](N1){$1$};
\path (N3) +(1.3,-.6) node [block](N2){$2$};
\path (N3) +(0,-1.2) node [block](N4){$4$};
\draw [->] (N1) -- (N3);
\draw [->] (N1) -- (N4);
\draw [->] (N2) -- (N3);
\draw [->] (N2) -- (N4);
\end{tikzpicture}
\end{center}
\end{minipage}
\begin{minipage}{0.3\linewidth}
\begin{center}
$\mathcal{S}_2 = \bmat{1&0&0&0\\0&1&0&0\\1&1&1&0\\1&1&0&1}$
\end{center}
\end{minipage}
\caption{A four-node graph and its adjacency matrix.\label{fig:ex1}}
\end{figure}

Our main result is that $\mathcal{S}$-realizability is necessary for $\mathcal{S}$-stabilizability; finding a stabilizing controller that has the same $\mathcal{S}$-structure as the plant. Non-realizable systems exist, but such systems cannot stabilize or be stabilized by other structured systems.

The paper is organized as follows. In Section~\ref{sec:literature}, we cover some related work in the literature and in Section~\ref{sec:preliminaries} we cover basic definitions and concepts touching on systems over graphs, realizability, and stabilizability. Our main results are in Sections~\ref{sec:simple} and \ref{sec:main}, which are followed by concluding remarks in Section~\ref{sec:discussion}.

\section{Literature Review}\label{sec:literature}

There is a large body of work exploring control systems defined over graphs. For a broad class of systems, synthesizing an optimal controller can be reduced to solving a convex optimization problem \cite{lessard_iqi_09,qimurti04,rotkowitz2010convexity}.

Several solution approaches have been reported. LMI methods \cite{dulldand_hetero,rantzer06,scherer02,zelazo09} work directly with state-space realizations for the plant and controller. Vectorization~\cite{rotkowitz_2006a} avoids the sparsity constraint by reshaping the transfer matrix. Alternatively, one can solve a sequence of finite-dimensional convex problems whose solutions converge to the optimal structured controllers~\cite{qimurti04}.

Surprisingly, the issue of structured realizability is not addressed in any of the aforementioned works. Indeed, the LMI methods assume that structured realizations for the plant and controller always exist, while transfer function methods make no mention of state-space --- so it's conceivable that a transfer function method might generate a controller which has no structured realization! 

In this paper, we consider directed graphs with delay-free links. Our framework is similar to that of Swigart and Lall~\cite{swigart_graph}. For many such problems, explicit state-space solutions have been directly computed \cite{jonghan_separable,lessard_bcof,lessard_allerton,shah10,swigart_thesis}. In these works, a structured realization is assumed for the plant, and the optimal controller turns out to have a structured realization as well. The topic of whether such realizations should exist in general, or whether suboptimal structured stabilizing controllers should always have structured realizations is not discussed.

To the best of our knowledge, the only existing work dealing with structured realizability is the work of Vamsi and Elia \cite{elia_cdc,elia_acc}. In these papers, the authors provide sufficient conditions for structured realizability, as well as an LMI approach for controller synthesis that guarantees realizability. However, they give no example of a provably non-realizable system, and the sufficient conditions provided are potentially very restrictive.

In Section~\ref{sec:intro}, we showed that structured realizability is a meaningful concept by giving an example of a non-realizable transfer function. In the sections that follow, we provide some very general results; a parameterization of all structured stabilizing controllers, and a proof that $\mathcal{S}$-realizability is necessary for $\mathcal{S}$-stabilizability. 

\section{Preliminaries}\label{sec:preliminaries}

\subsection{Directed Graphs}

A directed graph is a set of nodes $\mathcal{V}=\{1,\dots,N\}$ and edges $\mathcal{E} \subseteq \mathcal{V}\times\mathcal{V}$. If $(i,j)\in \mathcal{E}$, we say that there is an edge from $i$ to $j$, and we write $i\to j$. We make several assumptions regarding the graph.
\begin{enumerate}[{\it A}1)]
\item Self-loops: for all $i\in \mathcal{V}$, $i\to i$.
\item Transitive closure: if $i \to j$ and $j \to k$, then $i \to k$.
\item There are no directed cycles of length 2 or greater.
\end{enumerate}
These assumptions are natural, and will be further justified in Section~\ref{sec:sysgraphs}. Given a graph satisfying the assumptions above, define the adjacency matrix $\mathcal{S}\in \{0,1\}^{N\times N}$
\[
\mathcal{S}_{ij} = \begin{cases} 1 & j \to i \\ 0 & \text{otherwise} \end{cases}
\]
Under assumptions {\it A}1--{\it A}3, one can always relabel the nodes such that $\mathcal{S}$ is lower-triangular with a full diagonal. Therefore, we will assume a lower-triangular $\mathcal{S}$ from now on. See Figure~\ref{fig:ex1} for an example of a graph and its associated adjacency matrix.

\subsection{Systems over Graphs}\label{sec:sysgraphs}

The systems considered in this paper are linear, time-invariant, continuous-time, and rational. We denote the set of proper rational transfer functions as $\Rp$. If all the poles of  $G\in\Rp$ are contained in the open left-half plane, $G$ is \emph{stable}, and we write $G\in \RHinf$.

For systems defined over graphs, additional notation is needed to keep track of the input and output partitions.

\begin{defn}[Index sets]
Suppose we have a graph $\mathcal{S}\in\{0,1\}^{N\times N}$. An \emph{index set} $k$ is a tuple $(k_1,\dots,k_N)$ of nonnegative integers. We also define the set of nonempty indices as $\neset_k = \{i \in \mathcal{V}\mid k_i \ne 0\}$.
\end{defn}

\begin{defn}
Suppose we have a graph $\mathcal{S}\in\{0,1\}^{N\times N}$ with associated index sets $k$ and $m$, and $F$ is a commutative ring. We write $A \in \mathcal{S}(F,k,m)$ to mean that
\[
A = \bmat{ A_{11} & \cdots & A_{1N} \\
                 \vdots & \ddots & \vdots \\
                 A_{N1} & \cdots & A_{NN}}
\]
If $j\to i$, then $A_{ij} \in F^{k_i\times m_j}$. Otherwise, $A_{ij} = 0$. If $k$ and $m$ are clear by context, we simply write $A \in \mathcal{S}(F)$.
\end{defn}

A plant $G$ defined over the graph $(\mathcal{V,E})$ is written in its most general form as $G \in \mathcal{S}(\Rp,k,m)$. Intuitively, this means that if $i\to j$, then input $u_i$ affects output $y_j$, and the associated transfer function is $G_{ij}$. If there is no edge from $i$ to $j$, then $G_{ij}=0$.

We seek controllers that obey the same structure. Namely, if $i\to j$ then the control signal $u_j$ may depend on the measurement $y_i$, and the associated transfer function is $K_{ij}$. If there is no edge from $i$ to $j$ then $K_{ij}=0$.

We can now see why Assumptions~{\it A}2--{\it A}3 make sense. The graph $(\mathcal{V,E})$ represents information flow; if a subcontroller $i$ shares what it knows with subcontroller $j$ along the link $i\to j$ and similarly for $j \to k$, one would expect subcontroller $k$ to have access to the information from subcontroller $i$ as well. Directed cycles can also be removed by treating all nodes involved as a single node.

\begin{rem}\label{rem:algebra_property}
Assumption~{\it A}2 leads to a useful algebraic property.
If $G_1 \in \mathcal{S}(\Rp,k,m)$ and $G_2 \in \mathcal{S}(\Rp,m,p)$ then
$G_1G_2 \in \mathcal{S}(\Rp,k,p)$. In other words, structured transfer functions as defined above form an \emph{algebra}.
\end{rem}

\subsection{Structured Realizability}\label{sec:realizability}

Structured realizability is a core concept in this paper. Roughly speaking, we are interested in finding a state-space realization for a structured plant comprised of matrices that are also structured.

\begin{defn}[Structured realizability]
A transfer function $G\in\mathcal{S}(\Rp,k,m)$ is said to be $\mathcal{S}(k,m)$--\emph{realizable} if there exists an index set $n$ and matrices $A\in \mathcal{S}(\R,n,n)$, $B\in \mathcal{S}(\R,n,m)$, $C\in\mathcal{S}(\R,k,n)$, and $D\in\mathcal{S}(\R,k,m)$ such that $(A,B,C,D)$ is a stabilizable and detectable realization for $G$. When $m$, and $k$ are clear by context, we will simply write that $G$ is $\mathcal{S}$-realizable.
\end{defn}

Note that we allow the index set $n$ to have zero-entries. If a component of $n$ is zero, it means that no state is associated with that subsystem, and the corresponding rows and columns of $A$ collapse. 

A key aspect of this definition is the requirement of stabilizability and detectability. Indeed, we will see in Section~\ref{sec:stable} that one can always construct a trivial realization with the correct structure by introducing duplicate (non-minimal) states.

\subsection{Stabilizability}

Throughout this paper, we use the conventional notion of internal stability for feedback interconnections \cite{dullerud,zdg}. If a plant $G$ and controller $K$ are connected in feedback as in Figure~\ref{fig:feedback_stab}, then $K$ is stabilizing if and only if the map $(u_1,u_1)\mapsto (y_1,y_2)$ is proper and stable. When dealing with decentralized systems, there are other natural ways to define stability. For a further discussion on this topic, see Section~\ref{sec:discussion}. We now state the formal input-output definition of stabilization.
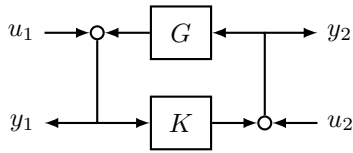
\begin{figure}[h]
\centering
\begin{tikzpicture}[thick,auto,>=latex,node distance=1.2cm]
\tikzstyle{block}=[draw,rectangle,inner sep=2mm,minimum width=0.8cm,minimum height=0.7cm]
\tikzstyle{plus}=[draw,circle,inner sep=0.4ex]
\node [block](G){$G$};
\node [block,below of=G](K){$K$};
\draw [->] (G.west) -- +(-0.6,0) node[plus,anchor=east](p1){};
\draw [->] (G.west)+(-1.4,0) node [anchor=east](u1){$u_1$} -- (p1);
\draw [->] (p1 |- K) -- (u1.east |- K) node [anchor=east](y1){$y_1$};
\draw [->] (p1) |- (K);
\path (G.east)+(0.6,0) node (tmp){};
\draw [->] (K) -- (tmp |- K) node[plus,anchor=west](p2){};
\draw [<-] (G.east) -| (p2);
\path (G.east)+(1.4,0) node [anchor=west](y2){$y_2$};
\draw [<-] (p2) -- (y2.west |- K) node [anchor=west](u2){$u_2$};
\draw [<-] (y2) -- (u1 -| p2);
\end{tikzpicture}
\caption{Feedback loop with inputs and outputs added to the feedback path.\label{fig:feedback_stab}}
\end{figure}
\begin{defn}[Stabilization]\label{def:stab_io}
Suppose $G \in \Rp^{k\times m}$ and $K\in\Rp^{m\times k}$. We say that $K$ \emph{stabilizes} $G$ if
\begin{enumerate}[(i)]
\item $I-G(\infty)K(\infty)$ is invertible, and
\item $\bmat{I & -G \\ -K & I}^{-1} \in \RHinf$
\end{enumerate}
\end{defn}
There is also a useful state-space characterization of stabilization, which we state as a proposition.
\begin{prop}\label{prop:stab_ss}
Suppose $G \in \Rp^{k\times m}$, $K\in\Rp^{m\times k}$ have
realizations given by $(A,B,C,D)$ and $(A_K,B_K,C_K,D_K)$ respectively. Then the following are equivalent.
\begin{enumerate}[(i)]
\item $(C,A,B)$ and $(C_K,A_K,B_K)$ are stabilizable and detectable, and $K$ stabilizes $G$.
\item  $(I-DD_K)$ is invertible, and $\bar A$ is Hurwitz, where
\end{enumerate}
\[
\bar A = 
\bmat{A & 0 \\ 0 & A_K}+
      \bmat{B & 0 \\ 0 & B_K}
      \bmat{I & -D_K \\ -D & I}^{-1}
      \bmat{0 & C_K \\ C & 0}.
\]
\end{prop}
\begin{proof} Proposition~\ref{prop:stab_ss} is a standard result~\cite{dullerud,zdg}, though it is typically stated with the assumption that $(C,A,B)$ is stabilizable and detectable. To prove the converse, note that if $\bar A$ is Hurwitz, then
\[
\left( \bmat{A & 0 \\ 0 & A_K},\bmat{B & 0 \\ 0 & B_K} \right)
\qquad\text{is stabilizable.}
\]
It follows from the PBH test that $(A,B)$ is stabilizable as well. A similar argument holds for detectability.
\end{proof}

In this paper, we seek controllers that are both stabilizing and have a particular structure. Therefore, we introduce new terminology to indicate this more restricted notion of stabilization.

\begin{defn}
Suppose $G \in \mathcal{S}(\Rp,k,m)$. We say that $G$ is $\mathcal{S}$-stabilizable if there exists $K \in \mathcal{S}(\Rp,m,k)$ such that $K$ stabilizes $G$.
\end{defn}

Note that these definitions are symmetric; $K$ stabilizes $G$ if and only if $G$ stabilizes $K$. The same symmetry relationship holds for $\mathcal{S}$-stabilization.

In the absence of structural constraints, it is well-known that every $G\in\Rp$ can be stabilized. The main thrust of this paper is to explain what happens when $G$ is structured, and we seek a controller that is $\mathcal{S}$-stabilizing. We will see in Sections~\ref{sec:simple}--\ref{sec:main} that not every $G\in\mathcal{S}$ is $\mathcal{S}$-stabilizable, and $\mathcal{S}$-realizability plays an important role.

\section{Simple Cases}\label{sec:simple}

For certain systems, finding a structured realization is straightforward. In this section, we explore two such cases that will be useful later: stable systems, and linear chain structures.

\subsection{Stable Systems}\label{sec:stable}

If the plant is stable, it can always be $\mathcal{S}$-realized, regardless of the underlying graph. By duplicating states in a way that guarantees the correct structure, the resulting realization is always stabilizable and detectable since the starting plant was stable. We give the construction in the following lemma. A similar result appeared in~\cite{elia_acc}.

\begin{lem}\label{lem:realize_stable}
Suppose $G \in \mathcal{S}(\RHinf,k,m)$. Then $G$ is $\mathcal{S}$-realizable.
\end{lem}
\begin{proof}
We may construct a realization for $G$ as follows. Separate $G$ into its block-columns $G_i$ for $i=1,\dots,N$. Find minimal realizations
\[
G_i = \stsp{A_i}{B_i}{C_i}{D_i}.
\]
In general, $A_i$ and $B_i$ will be full, but $C_i$ and $D_i$ will have zero-rows corresponding to the zero-rows in $G_i$. Now stack the columns side-by-side and obtain a joint realization.
\begin{equation}\label{eq:joint_real}
\bmat{\stsp{A_1}{B_1}{C_1}{D_1} &\cdots& \stsp{A_N}{B_N}{C_N}{D_N}} = 
\left[\begin{array}{ccc|ccc}
A_1 & & & B_1 & & \\
& \ddots & &  & \ddots & \\
& & A_N & & & B_N \\ \hlinet
C_1 & \cdots & C_N & D_1 & \cdots & D_N
\end{array}\right]
\end{equation}
The realization \eqref{eq:joint_real} has the desired structure because $A$ and $B$ are block-diagonal, and $C$ and $D$ are in $\mathcal{S}(\R)$. Finally, the realization is stabilizable and detectable since we started with a stable system and each $A_i$ is stable.
\end{proof}

In the case where $G$ is not stable, the construction method used in the proof of  Lemma~\ref{lem:realize_stable} still produces a realization with the correct structure, but the realization may fail to be stabilizable and detectable. This is a consequence of the fact that realizing the block-columns separately and then re-combining them may cause some unstable poles to get duplicated. Indeed, this is precisely what happens if we attempt to realize \eqref{eq:ex1}. 

\subsection{Chain Structures}\label{sec:triangular}

Linear chain structures correspond to adjacency matrices whose lower-triangular part is full. It turns out such systems are always $\mathcal{S}$-realizable.

\begin{lem}\label{lem:realize_triangular}
Suppose $G\in\mathcal{S}(\Rp,k,m)$ and $\mathcal{S}$ has the full lower-triangular sparsity pattern
\[
\mathcal{S} = \bmat{1 & & 0 \\ \vdots & \ddots & \\ 1 & \cdots & 1}
\]
Then $G$ is $\mathcal{S}$-realizable.
\end{lem}
\begin{proof}
Let $(A,B,C,D)$ be a minimal realization of $G$. We will sequentially construct a state transformation matrix $T$ such that $(T^{-1}AT, T^{-1}B, CT, D)$ has the desired sparsity pattern. Consider the simplest case, $N=2$. Partition the realization according to the block-structure
\begin{equation}\label{eq:st2}
\bmat{G_{11} & 0 \\ G_{21} & G_{22} }
= \left[\begin{array}{c|cc}
A & B_1 & B_2 \\ \hlinet
C_1 & D_{11} & 0 \\
C_2 & D_{21} & D_{22} \end{array}\right].
\end{equation}
Note that $D$ must already have the correct sparsity pattern.
A realization for the zero-block $G_{12}$ is
\begin{equation}\label{eq:st1}
0 = \left[\begin{array}{c|cc}
A & B_2 \\ \hlinet
C_1 & 0 \end{array}\right].
\end{equation}
Let $T_1$ be the transformation matrix that puts \eqref{eq:st1} into Kalman canonical form. There are typically four blocks in such a decomposition, but there will only be three in this case since the system we are realizing is identically zero and thus there can be no modes that are both controllable and observable. Applying $T_1$ to~\eqref{eq:st2}, we obtain
\[
G = \left[\begin{array}{ccc|cc}
A_{\bar c o} & 0 & 0 & B_{11} & 0\\
A_{21} & A_{\bar c \bar o} & 0 & B_{21} & 0 \\
A_{31} & A_{32} & A_{c\bar o} & B_{31} & B_{c \bar o} \\ \hlinet
C_{\bar c o} & 0 & 0 & D_{11} & 0\\
C_{21} & C_{22} & C_{23} & D_{21} & D_{22}
\end{array}\right]
\]
This realization is block-lower-triangular, and we notice that
the index sets $n$ may not be unique. For example, the modes $A_{\bar c \bar o}$ can be part of either the $A_{11}$ block or the $A_{22}$
block.

For the case where $N>2$, put the first block of the diagonal into $G_{11}$ and lump the rest of the block-lower-triangular structure into $G_{22}$. Upon applying the $T_1$ found from the $N=2$ case, we are left with
\[
G_{22} = \stsp{A_{c\bar o}}{B_{c\bar o}}{C_{23}}{D_{22}}.
\]
Apply this approach recursively by finding a $T_2$ that puts the first zero-block-row of $G_{22}$ into Kalman canonical form. Continuing in this manner eventually yields a realization for $G$ that has the desired sparsity pattern. Furthermore, this realization is stabilizable and detectable since it is minimal.
\end{proof}

\section{Main Results}\label{sec:main}

Our main results draw the connection between $\mathcal{S}$-stabilizability and $\mathcal{S}$-realizability. We will show that $\mathcal{S}$-stabilizable plants are always $\mathcal{S}$-realizable, and $\mathcal{S}$-stabilizing controllers are always $\mathcal{S}$-realizable. In other words, if structured systems connected in feedback results in a an internally stable interconnection, then both systems must be $\mathcal{S}$-realizable.

The main result has two essential ingredients. In Section~\ref{sec:param}, we assume the plant is $\mathcal{S}$-realizable. We then find a Youla-like parameterization of all $\mathcal{S}$-stabilizing controllers and show that all $\mathcal{S}$-stabilizing controllers are $\mathcal{S}$-realizable. In Section~\ref{sec:stabilizable}, we assume instead that our plant is $\mathcal{S}$-stabilizable, but make no assumptions regarding $\mathcal{S}$-realizability. We then show that there must exist a $\mathcal{S}$-stabilizing controller that is also $\mathcal{S}$-realizable. These results are combined into one concise statement, which we give in Corollary~\ref{cor:main}.

\subsection{Parameterization of Stabilizing Controllers}\label{sec:param} 

In this subsection, we assume the plant is $\mathcal{S}$-realizable, and we give a characterization of all structured stabilizing controllers. This parameterization is similar to the classical work of Youla \cite{youla}, and the unstructured version is one of the pillars of classical $\Htwo/\mathcal{H}_\infty$ theory~\cite{zdg}. In the framework adopted herein, the sparsity constraint imposed on the controller amounts to the same sparsity constraint being imposed on the Youla parameter~$Q$.

Similar parameterizations have appeared in the literature. Most notably, the work of Qi et.al. \cite{qimurti04} treat the lower-triangular case and other cases of interest. A more general parameterization is provided in \cite{sabau10}, where the authors consider general quadratically invariant systems in which the plant and controller may have different sparsity patterns.

\begin{thm}\label{thm:conditions}
Suppose $G \in \mathcal{S}(\Rp,k,m)$ is $\mathcal{S}$-realizable, with a structured realization given by $(A,B,C,D)$.
\begin{enumerate}[(i)]
\item $G$ is $\mathcal{S}$-stabilizable if and only if $(C_{ii},A_{ii},B_{ii})$ is stabilizable and detectable for all $i\in\neset_n$.
\item In this case, for $i\in\neset_n$, choose $F_i$ and $L_i$ such that $A_{ii}+B_{ii}F_i$ and $A_{ii}+L_iC_{ii}$ are Hurwitz. Define $F_d = \diag\{F_i\}_{i\in \neset_n}$ and $L_d = \diag\{L_i\}_{i\in \neset_n}$. A particular $\mathcal{S}$-stabilizing controller is given by
\[
K_0 = \stsp{A+BF_d+L_dC+L_dDF_d}{-L_d}{F_d}{0}
\]
\end{enumerate}
\end{thm}
\begin{proof}
Suppose $(C_{ii},A_{ii},B_{ii})$ is stabilizable and detectable for all $i\in \neset$. By definition, there must exist $F_i$ and $L_i$ such that $A_{ii}+B_{ii}F_i$ and $A_{ii}+L_iC_{ii}$ are Hurwitz. Define $F_d = \diag\{F_i\}_{i\in \neset_n}$ and $L_d = \diag\{L_i\}_{i\in \neset_n}$. Since $A$ is block-lower-triangular, $A+BF_d$ and $A+L_dC$ are also Hurwitz. It is straightforward to check that $K_0$ defined in the theorem statement above is a stabilizing controller with the correct structure.

Conversely, suppose $G$ is stabilizable by some $K \in \mathcal{S}(\Rp,m,k)$. Let $\bar{\mathcal{S}}$ be the adjacency matrix with a full lower-triangular sparsity 
pattern. Then $K \in \bar{\mathcal{S}}(\Rp,m,k)$. By Lemma~\ref{lem:realize_triangular}, $K$ is $\bar{\mathcal{S}}$-realizable. Let 
$(A_K,B_K,C_K,D_K)$ be a stabilizable and detectable realization for $K$ structured according to $\bar{\mathcal{S}}$. By Proposition~\ref{prop:stab_ss}, $(I-DD_K)$ must be invertible, and $\bar A$ must be Hurwitz. Rearrange the columns and rows of $\bar A$ by taking the first sub-blocks of each large block, then the second sub-blocks, and so on. The resulting matrix is block-lower-triangular, and has the same eigenvalues as $\bar A$. The main diagonal blocks are
\[
\bmat{A_{ii} & 0 \\ 0 & A_{Kii}}+
      \bmat{B_{ii} & 0 \\ 0 & B_{Kii}}
      \bmat{I & \!\!-D_{Kii} \\ -D_{ii} & \!\!I}^{-1} 
      \bmat{0 & C_{Kii} \\ C_{ii} & 0}\quad\text{for }i\in\neset_n
\]
Applying Proposition~\ref{prop:stab_ss} once more, we conclude that $(C_{ii},A_{ii},B_{ii})$ is stabilizable and detectable.
\end{proof}

Theorem~\ref{thm:conditions} establishes a necessary and sufficient condition for the existence of a $\mathcal{S}$-stabilizing controllers for for $\mathcal{S}$-realizable plants. Our next result is a parameterization of all $\mathcal{S}$-stabilizing controllers.

\begin{thm}\label{thm:param}
Suppose $G \in \mathcal{S}(\Rp,k,m)$ is $\mathcal{S}$-realizable, with a structured realization given by $(A,B,C,D)$. Let
    \begin{equation*}\label{eq:J}
      J = 
      \left[\begin{array}{c|cc}
          A+BF_d+L_dC + L_d DF_d& -L_d & B+L_d D \\ \hlinet
          F_d & 0 & I \\
          -(C+DF_d) & I & -D
        \end{array}\right]
    \end{equation*}
where $F_d$ and $L_d$ are defined as in Theorem~\ref{thm:conditions}. Also, define \\
$F_l(J,Q) = J_{11}+J_{12}Q(I-J_{22}Q)^{-1}J_{21}$.
\begin{enumerate}[(i)]
\item Every $K\in\mathcal{S}(\Rp,m,k)$ that stabilizes $G$ is parameterized by $K = F_l(J,Q)$, where $Q \in \mathcal{S}(\RHinf,m,k)$ such that $I+DQ(\infty)$ is nonsingular.
\item Every stabilizing controller is $\mathcal{S}$-realizable.
\end{enumerate}
\end{thm}
\begin{proof}
By the proof of Theorem~\ref{thm:conditions}, $A+BF_d$ and $A+L_dC$ are Hurwitz. The classical parameterization of all stabilizing controllers~\cite{zdg} is given by $F_l(J,Q)$ where $J$ is defined in the theorem statement above, and $Q\in \RHinf$ such that $I+DQ(\infty)$ is nonsingular. It remains to show that $K\in \mathcal{S}$ if and only if $Q\in \mathcal{S}$. This is a straightforward consequence of Remark~\ref{rem:algebra_property}, and the fact that all subblocks of $J$ are in $\mathcal{S}$.

We will now show that every stabilizing controller is $\mathcal{S}$-realizable. Note that $Q\in\mathcal{S}(\RHinf,m,k)$, so we may apply Lemma~\ref{lem:realize_stable}. Let $(A_Q,B_Q,C_Q,D_Q)$ be a structured realization for $Q$ that is stabilizable and detectable. The associated controller has a realization 
\[
F_l(J,Q)=\left[\begin{array}{cc|c}
\hat A_{11} & \hat A_{12} & \hat B_1 \\
\hat A_{21} & \hat A_{22} & \hat B_2 \\ \hlinet\rule{0pt}{2.5ex}
\hat C_1 & \hat C_2 & \hat D \end{array}\right]
\]
which can be computed using the Redheffer Star-Product \cite[\S 10]{zdg}. For example, the $\hat A$ and $\hat B$ terms are:
\begin{align*}
\hat A_{11} &= A+BF_d+L_dC+L_dDF_d -(B+L_dD)(I+D_QD)^{-1} (C+DF_d) \\
\hat A_{12} &= (B+L_dD)(I+D_QD)^{-1} C_Q \\
\hat A_{21} &= -B_Q(I+DD_Q)^{-1}(C+DF_d) \\
\hat A_{22} &= A_Q -B_Q(I+DD_Q)^{-1}DC_Q \\
\hat B_1 &= -L_d +(B+L_dD)(I+D_QD)^{-1}D_Q \\
\hat B_2 &= B_Q(I+DD_Q)^{-1}
\end{align*}
Using Remark~\ref{rem:algebra_property} and the fact that the realization for $Q$ is structured, it is straightforward to show that the $\hat A_{ij}$ and $\hat B_i$ are each structured according to $\mathcal{S}$. The same is true of the $\hat C_i$ and $\hat D$. Stabilizability can be verified explicitly. For example, one can check that
\[
\bmat{\hat A_{11} & \hat A_{12} \\ \hat A_{21} & \hat A_{22}} +
\bmat{\hat B_1 \\ \hat B_2} \bmat{C+DF_d & DC_Q}
= \bmat{ A+BF_d & BC_Q \\ 0 & A_Q }
\]
which is clearly Hurwitz. A similar argument applies for detectability. Finally, we can permute the states as we did in the proof of Theorem~\ref{thm:param} to obtain a realization for $F_l(J,Q)$ for which all matrices are in $\mathcal{S}$. This shows that the controller is $\mathcal{S}$-realizable, as required.
\end{proof}

\subsection{Stabilizable Systems are Realizable}\label{sec:stabilizable}

In this subsection, we assume the plant is $\mathcal{S}$-stabilizable, but make no assumptions regarding $\mathcal{S}$-realizability. In this case, we show that one can always construct an $\mathcal{S}$-realizable controller.

\begin{thm}\label{thm:diagctrl}
Suppose $G \in \mathcal{S}(\Rp,k,m)$ is stabilized by some $K \in \mathcal{S}(\Rp,m,k)$. Then,
\begin{enumerate}[(i)]
\item $G$ is stabilized by $K_d = \diag(K_{11},\dots,K_{N\!N})$.
\item $K_d$ is $\mathcal{S}$-realizable.
\end{enumerate}
\end{thm}
\begin{proof}
The proof is similar to the proof of Theorem~\ref{thm:conditions}. $G$ and $K$ may not be $\mathcal{S}$-realizable, so we embed their sparsity patterns in $\bar{\mathcal{S}} \supset \mathcal{S}$, the set of block-lower-triangular transfer matrices. Apply Lemma~\ref{lem:realize_triangular} to find triangular realizations, and the result follows from  Proposition~\ref{prop:stab_ss}.
\end{proof}

Theorem~\ref{thm:diagctrl} provides a simple test for $\mathcal{S}$-stabilizability; stabilizing each subsystem individually must also globally stabilize the system. Otherwise, no $\mathcal{S}$-stabilizing controller exists.

\begin{cor}\label{cor:diagtest}
Suppose $G \in \mathcal{S}(\Rp,k,m)$. For $i=1,\dots,N$, let $K_{i}\in \Rp^{m_i\times k_i}$  be any controller that stabilizes $G_{ii}$. Define $K_d = \diag(K_{1},\dots,K_{N})$. Then $G$ is $\mathcal{S}$-stabilizable if and only if $K_d$ stabilizes $G$.
\end{cor}
\begin{proof}
If $K_d$ stabilizes $G$, then clearly $G$ is $\mathcal{S}$-stabilizable, since $K_d\in\mathcal{S}$. The converse follows from Theorem~\ref{thm:diagctrl}.
\end{proof}

The main result of this paper is that stabilizable plants are always $\mathcal{S}$-realizable. We state it below as a Corollary.

\begin{cor}[Main Result]\label{cor:main}
Suppose $G\in\mathcal{S}(\Rp,k,m)$ is $\mathcal{S}$-stabilizable. Then $G$ is $\mathcal{S}$-realizable, and every $K\in\mathcal{S}(\Rp,m,k)$ that stabilizes $G$ is $\mathcal{S}$-realizable.
\end{cor}
\begin{proof}
By Theorem~\ref{thm:diagctrl}, every $\mathcal{S}$-stabilizable plant $G$ can be stabilized by an $\mathcal{S}$-realizable controller $K_d$. Interchanging the roles of the plant and controller, $K_d$ is an $\mathcal{S}$-realizable plant stabilized by $G$. By Theorem~\ref{thm:param}, $G$ must be $\mathcal{S}$-realizable. Applying Theorem~\ref{thm:param} once more to $G$, we conclude that all stabilizing controllers must be $\mathcal{S}$-realizable.
\end{proof}

Note that the converse of Corollary~\ref{cor:main} is generally false. One can construct an $\mathcal{S}$-realizable plant that is not $\mathcal{S}$-stabilizable. Consider for example:
\[
G = \bmat{\frac{1}{s+1} & 0 \\ \T \frac{1}{s-1} &
  \frac{1}{s+1}} = \left[\begin{array}{ccc|cc}
    -1 & 0 & 0 & 1 & 0 \\
    0 & 1 & 0 & 1 & 0 \\
    0 & 0 & -1 & 0 & 1 \\ \hlinet
    1 & 0 & 0 & 0 & 0 \\
    0 & 1 & 1 & 0 & 0
\end{array}\right]
\]
If we realize $G$ by clustering the unstable mode with $A_{11}$, i.e. using the index set $n=(2,1)$, it leads to an undetectable $(C_{11},A_{11})$. If we instead associate the unstable mode with $A_{22}$ by using $n=(1,2)$, it will lead to an unstabilizable $(A_{22},B_{22})$. By Theorem~\ref{thm:conditions}, this plant is not stabilizable by a lower-triangular controller. Of course, non-structured stabilizing controllers exist since the realization for $G$ given above is minimal.

\section{Concluding Remarks}\label{sec:discussion} 

The main contribution of this paper is to establish the connection between the structured versions of stabilizability and realizability. In the context of structured controller synthesis, it is a basic requirement that the plant is $\mathcal{S}$-stabilizable and that the controller is $\mathcal{S}$-stabilizing. By Corollary~\ref{cor:main}, it is necessary that the plant and controller are both $\mathcal{S}$-realizable. The counterexample~\eqref{eq:ex1} causes no problems because as a plant, $G_2$ would not be $\mathcal{S}$-stabilizable, and as a controller, $G_2$ could not stabilize any structured plant.

As mentioned in Section~\ref{sec:literature}, the work of Vamsi and Elia~\cite{elia_acc,elia_cdc} also addresses the issue of structured realizability. The underlying assumptions are different however. The authors assume discrete-time subsystems with one-timestep delays along directed edges. In the work herein, we assumed a delay-free network of continuous-time systems. Some results, such as the counterexample \eqref{eq:ex1}, are valid in either framework. However, the other results in this paper do not immediately translate to discrete-time.

In this paper, we adopted the classical definition of internal stability. This notion of stabilization is the weakest one possible, since it does not guarantee that the signals that travel between subsystems or between subcontrollers remains bounded. It follows that Corollary~\ref{cor:main} will remain true even for stricter notions of stabilization. The work of Yadav et.al.~\cite{voulgaris_structure} explores other types of stability in great detail.

While $\mathcal{S}$-stabilizability guarantees $\mathcal{S}$-realizability, the question of how to explicitly construct a structured realization is still open. Lemmas~\ref{lem:realize_stable} and \ref{lem:realize_triangular} provide constructions for the special cases of stable and triangular systems respectively, but no universal method for constructing structured realizations is known.

\bibliographystyle{abbrv}
\bibliography{realizability}

\appendix

\section{Proof of Counterexample}\label{sec:appendix}

In this section, we provide a proof that the example \eqref{eq:ex1} of Figure~\ref{fig:ex1} is not $\mathcal{S}_2$-realizable.

\begin{proof}
Suppose $G_2$ has a stabilizable and detectable realization in $\mathcal{S}_2$. Label the blocks of this realization as
\[
A = \bmat{A_{11} &   0    &   0    &   0 \\
            0    & A_{22} &   0    &   0 \\
          A_{31} & A_{32} & A_{33} &   0 \\
          A_{41} & A_{42} &   0    & A_{44} }.
\]
and similarly for $B$ and $C$. Expand $C(sI-A)^{-1}B$ and equate with $G_2$. The four equations corresponding to each of the nonzero entries $G_{ij}$ are
\begin{equation}\label{eq:counter}
\bmat{C_{ij} & C_{ii}}
\left( sI - \bmat{A_{jj} & 0 \\ A_{ij} & A_{ii}}\right)^{-1}
\bmat{B_{jj} \\ B_{ij}}
= \frac{1}{s-1}
\end{equation}
for $(i,j) \in \big\{ (3,1), (3,2), (4,1), (4,2)\big\}$.
A minimal realization for $G_2$ is given by
\[
G_2 = 
\left[\begin{array}{c|cccc}
1 & 1 & 1 & 0 & 0 \\ \hlinet
0 & 0 & 0 & 0 & 0 \\
0 & 0 & 0 & 0 & 0 \\
1 & 0 & 0 & 0 & 0 \\
1 & 0 & 0 & 0 & 0 \end{array}\right]
\]
Since we have assumed our realization to be stabilizable and detectable, the eigenvalues of $A$ must consist of the unstable eigenvalue $1$, together with some number of stable eigenvalues. Since $A$ is block-lower-triangular, the  eigenvalue~$1$ must appear as an eigenvalue of exactly one of the $A_{ii}$ for $i=1,\dots,4$. For $(i,j)=(3,1)$, we deduce from~\eqref{eq:counter} that the eigenvalue~1 is contained in either $A_{11}$ or $A_{33}$, but for $(i,j)=(4,2)$, we deduce that it is contained in either $A_{22}$ or $A_{44}$. This contradiction implies that no structured stabilizable and detectable realization of $G_2$ can exist.
\end{proof}


\end{document}